\documentclass[12pt, a4paper]{article}

\usepackage{lipsum}
\usepackage{amsmath}
\usepackage{amssymb}
\usepackage{amsthm}
\usepackage{mathtools}
\usepackage{fancyhdr}
\usepackage{graphicx}
\usepackage{verbatim}
\usepackage{latexsym}
\usepackage[color = lightgray]{todonotes}
\usepackage{mathrsfs}
\usepackage{tikz}
\usepackage{tikz-cd}
\usepackage{float}

\usepackage{todonotes}

\usepackage{enumerate}
\usepackage{pdfpages}
\usepackage{xcolor}
\usepackage{sectsty}
\usepackage[calcwidth]{titlesec}
\usepackage{stmaryrd}
\usepackage{wallpaper}
\usepackage{setspace}
\usepackage[top = 2 cm, bottom = 3 cm, left = 2 cm, right = 2 cm]{geometry}
\usepackage{bbm}
\usepackage{datetime2}
\usepackage{anyfontsize}
\usepackage{pgf,pgfplots}
\usepackage{textcomp}
\usepackage[margin = 1cm]{caption}
\usepackage{stackengine}

\usepackage{wasysym}

\usepackage[style=numeric-comp,
  giveninits=true,
  isbn=false,
  maxbibnames=99]{biblatex}
\DeclareFieldFormat*{titlecase}{\MakeSentenceCase{#1}}

\DeclareSourcemap{
  \maps[datatype=bibtex]{
    \map[overwrite]{
      \step[fieldsource=doi, final]
      \step[fieldset=url, null]
      \step[fieldset=eprint, null]
    }  
  }
}

\AtEveryBibitem{%
  \clearlist{language}%
  \clearlist{month}%
  \ifentrytype{inproceedings}
    {\clearlist{booktitle}%
     \clearlist{publisher}%
     \clearlist{location}}
}

\renewbibmacro{in:}{}

\DeclareFieldFormat[misc]{title}{\mkbibquote{#1}}
\DeclareFieldFormat[article,inproceedings]{volume}{\mkbibbold{#1}}
\DeclareFieldFormat[inproceedings]{series}{\mkbibitalic{#1}}

\PassOptionsToPackage{hyphens}{url}\usepackage[hidelinks]{hyperref}
\usepackage[noabbrev]{cleveref}

\pgfplotsset{compat=1.16}
\usetikzlibrary{arrows}
\usetikzlibrary{arrows.meta}
\usetikzlibrary{patterns}

\usepackage{chngcntr}
\counterwithin{figure}{section}
\counterwithin{equation}{section}

\usepackage[T1]{fontenc}
\usepackage[cmbraces]{newtxmath}
\newcommand{\R}{\mathbb{R}}
\newcommand{\C}{\mathbb{C}}

\newcommand{\Z}{\mathbb{Z}}

\renewcommand{\P}{\mathbb{P}}

\renewcommand{\vec}[1]{\bar{\boldsymbol{#1}}}
\newcommand{\x}{\vec{x}}

\newtheoremstyle{theorems}% name
  {3pt}%      Space above
  {3pt}%      Space below
  {\itshape}%         Body font
  {}%         Indent amount (empty = no indent, \parindent = para indent)
  {\bfseries}% Thm head font
  {.}%        Punctuation after thm head
  { }%     Space after thm head: " " = normal interword space;
  {}%         Thm head spec (can be left empty, meaning `normal')

\newtheoremstyle{proofparts}% name
  {3pt}%      Space above
  {0pt}%      Space below
  {}%         Body font
  {\parindent}%         Indent amount (empty = no indent, \parindent = para indent)
  {\scshape}% Thm head font
  {:}%        Punctuation after thm head
  {\newline}%     Space after thm head: " " = normal interword space;
  {}%         Thm head spec (can be left empty, meaning `normal')

\newtheoremstyle{claims}% name
  {2pt}%      Space above
  {2pt}%      Space below
  {}%         Body font
  {\parindent}%         Indent amount (empty = no indent, \parindent = para indent)
  {\bfseries}% Thm head font
  {.}%        Punctuation after thm head
  { }%     Space after thm head: " " = normal interword space;
  {}%         Thm head spec (can be left empty, meaning `normal')

\theoremstyle{theorems}
\newtheorem{thm}{Theorem}[section]

\newtheorem{cor}[thm]{Corollary}
\newtheorem{prop}[thm]{Proposition}

\theoremstyle{definition}
\newtheorem{defn}[thm]{Definition}

\newtheorem{ex}[thm]{Example}
\newtheorem{remark}[thm]{Remark}

\theoremstyle{proofparts}

\makeatletter
\@addtoreset{proofpart}{thm}
\makeatother

\theoremstyle{claims}

\newtheorem*{claim*}{Claim}

\let\Sectionmark\sectionmark
\def\sectionmark#1{\def\Sectionname{#1}\Sectionmark{#1}}
\let\Subsectionmark\subsectionmark
\def\subsectionmark#1{\def\Subsectionname{#1}\Subsectionmark{#1}}

\newcommand{\abs}[1]{\left\vert #1 \right\vert}

\DeclareMathOperator{\diam}{diam}
\DeclareMathOperator{\supp}{supp}

\DeclareMathOperator{\tRe}{Re}

\DeclareMathOperator{\curl}{curl}
\let\div\relax\DeclareMathOperator{\div}{div}

\definecolor{emphcolor}{rgb}{0,0,1}           % EMPHASIS COLOUR

\newcommand{\ip}[2]{\left\langle #1 \middle\vert #2 \right\rangle}

\newcommand{\ud}{\,\textnormal{d}}
\newcommand{\dd}[1]{\frac{\textnormal{d}}{\textnormal{d} #1}}

\let\epsilon\varepsilon
\let\varepsilon\epsilon

\let\eps\epsilon

\title{Floating Wigner Crystal and Periodic Jellium Configurations}
\author{Asbjørn Bækgaard Lauritsen\footnote{
\stackunder{
\href{mailto:asbjornbaekgaard.lauritsen@ist.ac.at}{\nolinkurl{asbjornbaekgaard.lauritsen@ist.ac.at}}
}{
IST Austria (Institute of Science and Technology Austria), Am Campus 1, 3400 Klosterneuburg, Austria
}
}}

\begin{document}
\maketitle
\thispagestyle{empty}
\begin{abstract}
Extending on ideas of Lewin, Lieb and Seiringer (Phys Rev B, 100, 035127, (2019))
we present a modified ``floating crystal'' trial state for Jellium 
(also known as the classical homogeneous electron gas)
with density equal to a characteristic function. 
This allows us to show that three definitions of
the Jellium energy coincide in dimensions $d\geq 2$,
thus extending the result of Cotar and Petrache (arXiv: 1707.07664) and
Lewin, Lieb and Seiringer (Phys Rev B, 100, 035127, (2019))
that the three definitions coincide in dimension $d \geq 3$.
We show that the Jellium energy is also equivalent to a ``renormalized energy''
studied in a series of papers by Serfaty and others
and thus, by work of Bétermin and Sandier (Constr Approx, 47:39-74, (2018)),
we relate the Jellium energy to the order $n$ term in the logarithmic energy of $n$ points on the unit 2-sphere.
We improve upon known lower bounds for this renormalized energy.
Additionally, we derive formulas for the Jellium energy 
of periodic configurations.
\\ \\
{
	\bfseries
	Mathematics subject classification: 
}
82B05.
\\
{
	\bfseries
	Keywords:
}
Jellium, Statistical mechanics, Renormalized energy, Triangular lattice.
\end{abstract}
\section{Introduction}
The Jellium model is an important and very simple model, 
which models electrons in a uniformly charged background.
In this paper we 
discuss three different definitions of the Jellium energy density $e_\text{Jel}$.
These are known to 
coincide in dimensions $d \geq 3$ \cite{lewin.lieb.seiringer.19,cotar.petrache.2019}
and differ in dimension $d=1$ \cite{colombo.pascale.di-marino.2015,lewin.lieb.seiringer.19,choquard.75,kunz.74,baxter.63}.
We show that they coincide in dimension $d\geq 2$
using a similar method  as in \cite{lewin.lieb.seiringer.19}.
This verifies a conjecture by Cotar and Petrache \cite[Remark 1.7]{cotar.petrache.2019},
that the definitions of the Jellium energy density coincide in dimensions $d=2$.
The main difference in our argument as compared to that of \cite{lewin.lieb.seiringer.19} 
is the choice of a slightly different 
``floating crystal'' trial state, for which the density is a characteristic function. 
In dimension $d=3$ the thermodynamic limit of the Uniform Electron Gas 
exists under weaker conditions by the Graf-Schenker inequality \cite{graf.schenker.1995} as discussed in 
\cite{lewin.lieb.seiringer.18,lewin.lieb.seiringer.19}.
We do not have this stronger form of the thermodynamic limit in dimensions $d\ne 3$. 
Hence we need a trial state, for which the density is a characteristic function.

Secondly we consider 
the question of periodic Jellium configurations, where the electrons are confined to 
sites of a lattice.
Here we find that the energy is given by the Epstein (or lattice) $\zeta$-function
associated to the lattice, on which the electrons sit.

Thirdly we relate the evaluation of the Jellium energy to that of a ``renormalized energy'' $W$
studied in \cite{borodin.serfaty.13,sandier.serfaty.15,sandier.serfaty.12,rougerie.sefaty.16,petrache.serfaty.2017,betermin.sandier.18,leble.serfaty.17,cotar.petrache.2019}.
Here we show that $\min_{\mathcal{A}_1} W = 2\pi e_\text{Jel}$ (notation explained in \cref{sect.renorm.energy}). 
This gives another equivalent definition of the Jellium energy.
Bétermin and Sandier \cite{betermin.sandier.18} 
have shown that the logarithmic energy on $n$ points on $\mathbb{S}^2$ has a term of order $n$ given by
$c_\text{log} = \frac{1}{\pi} \min_{\mathcal{A}_1}W + \frac{\log 4\pi}{2}$,
(see \Cref{remark.log.s2}).
Hence we get yet another equivalent definition of the Jellium energy.
The relation of $e_\text{Jel}$ and $\min_{\mathcal{A}_1}W$ 
carry over the known bounds for the Jellium energy.
The lower bound of $e_\text{Jel} \geq -0.66118$ by Lieb and Narnhofer \cite{lieb.narnhofer.75} and Sari and Merlini \cite{sari.merlini.76}
has been known for many years. 
It improves upon known bounds for the constants $c_\text{log}$ and $\min_{\mathcal{A}_1} W$.
In particular it gives the bound $-0.0569 \leq c_\text{log} \leq -0.0556$
improving on the best known lower bound of $c_\text{log}\geq -0.0954$ 
due to Steinerberger \cite{steinerberger.20}.
Since the proof of the lower bound in \cite{sari.merlini.76} is not very detailed,
we give the proof in the appendix.

\section{Three Definitions of the Jellium Energy}
We now introduce the three models. 
We will give the argument only in dimension $d=2$, 
partly because this case is where the argument is most complicated
and partly because the physically interesting cases are dimensions $d=1,2,3$,
and the cases $d\geq3$ are solved \cite{lewin.lieb.seiringer.19,cotar.petrache.2019}.
For dimensions $d\geq 3$ the argument is the same, only one should replace every occurrence of $-\log$ with $|\cdot|^{2-d}$.

The first model is what we will call Jellium.
By scaling we may assume that the density of the background is $\rho = 1$.
Then the Jellium energy of $N$ particles in a domain $\Omega_N$ of size $|\Omega_N| = N$ 
is  
\[
\mathcal{E}_\textnormal{Jel}(\Omega_N, x_1, \ldots, x_N) 
	= 
		- \sum_{j<k} \log |x_j - x_k|
		+ \sum_{j=1}^N \int_{\Omega_N} \log |x_j - y| \ud y
		- \frac{1}{2} \iint_{\Omega_N \times \Omega_N} \log|x-y| \ud x\ud y
\]
The electrons are thought of as discrete classical particles in a uniform (positive) background, such that the entire system is neutral.
The electrons and background all interact through Coulomb interaction, in 2 dimensions given by $-\log|x|$. 
The long range behaviour of the logarithm means that this setting is somewhat different from the 3-dimensional case.
In \cite{sari.merlini.76} it is shown that the thermodynamical limit
\[
	e_\textnormal{Jel} = \lim_{\Omega_N \nearrow \R^2} \min_{x_1, \ldots, x_N \in \R^2} \frac{\mathcal{E}_\textnormal{Jel}(\Omega_N, x_1, \ldots, x_N)}{|\Omega_N|}
\]
exists under fairly non-restrictive conditions on the sequence of domains $\Omega_N$. 
For instance $\Omega_N = N^{1/2}\Omega$ for a fixed convex set $\Omega$ of size $|\Omega| = 1$.

The second model is that of periodic Jellium.
Here the $n = \ell^2$ electrons live on a torus of side length $\ell$ 
in a uniform background of opposite charge. 
The Coulomb potential between the electrons is replaced by the 
periodic Coulomb potential, where the electrons interact 
with all the mirror images of the other electrons and the uniform background.
The functional is defined as follows.

First, we define the periodic Coulomb potential $G_\ell$ as follows.
$G_\ell(x) = G_1(x/\ell)$, where $G_1$ is the one-periodic Coulomb potential,
satisfying $-\Delta G_1 = 2\pi \left(\sum_{z \in \Z^2} \delta_z - 1\right)$ and $\int_{C_1} G_1 \ud x= 0$, 
where $C_1 = (-1/2, 1/2)^2$.
It corresponds to the potential generated by a point charge and all its images
together with a uniform oppositely charged background. 
The background must be included for this not to diverge. 
Then
\[
	G_\ell(x) = G_1(x/\ell) = \frac{2\pi}{\ell^2} \sum_{\substack{k \in \frac{2\pi}{\ell}\Z^2 \\ k \ne 0}} \frac{1}{k^2} e^{ikx}.
\]
Now, $G_1(x) + \log |x|$ has a limit as $x \to 0$, which we call $C_\text{mad}$.
It is the Madelung constant, i.e. twice the energy per particle of 
the configuration with 1 particle in the unit cell - i.e. a square lattice configuration.
The functional $\mathcal{E}_{\text{per}, \ell}$ may now be defined as
\[
	\mathcal{E}_{\text{per}, \ell}(x_1, \ldots, x_n) = \sum_{j < k} G_\ell(x_j - x_k) + \frac{n}{2}\left(\log \ell + C_\text{mad}\right).
\]
The first term is what one gets if one just naively replaces the Coulomb interaction in the Jellium functional
by the periodic version $G_\ell$. Note that then the particle-background and background-background terms vanish
due to the fact that $\int_{C_\ell} G_\ell \ud x = 0$. 
We then define
\[
	e_\text{per} = \lim_{\ell \to \infty} \min_{x_1, \ldots, x_n} \frac{\mathcal{E}_{\text{per}, \ell}(x_1,\ldots,x_n)}{\ell^2}.
\]
The existence of this limit was established in \cite{petrache.serfaty.2017,sandier.serfaty.15,rougerie.sefaty.16,leble.serfaty.17,hardin.saff.simanek.su.2016}.
It will also follow from the proof of \Cref{thm.main} that indeed this limit exists.

The third model is what has been called 
the uniform electron gas (UEG) in \cite{lewin.lieb.seiringer.18,lewin.lieb.seiringer.19}.
For a complete description of this model see \cite{lewin.lieb.seiringer.18}.
Here there is no background charge, and the electrons are no longer point particles.
Instead the electrons are distributed according to a 
probability density $\P$ (meaning $\P$ is a probability measure on $\R^{2N}$)
which we require to give a constant density $\rho_\P = \mathbbm{1}_{\Omega_N}$,
where $\rho_\P$ is the sum of all the marginals.
The indirect energy of the distribution is then
\[
	\mathcal{E}_\textnormal{Ind}(\P) 
	= - \int \sum_{j < k} \log |x_j - x_k| \ud \P(x_1, \ldots, x_N)
		+ \frac{1}{2} \iint \log |x-y| \rho_\P(x) \rho_\P(y) \ud x \ud y.
\]
We are interested in keeping the density fixed, and so, for any density $\rho$ with $\int \rho \ud x = N$ we define
\[
	\mathcal{E}_{\textnormal{Ind}}(\rho) = \min_{\P: \rho_\P = \rho} \mathcal{E}_\textnormal{Ind}(\P)
\]
Since the electrons are indistinguishable, we should in principle 
restrict to symmetric $\P$'s. This however gives the same minimum.
Again, we are interested in the thermodynamic limit, and for a system of uniform density, i.e.
\[	
	e_\textnormal{UEG} = \lim_{\Omega_N \nearrow \R^2} \frac{\mathcal{E}_\textnormal{Ind}(\mathbbm{1}_{\Omega_N})}{|\Omega_N|}.
\]
The existence of this was established in \cite[Theorem 2.6]{lewin.lieb.seiringer.18}.
Their proof is done in dimensions $d \geq 3$, but works without change also in dimensions $d=1,2$.
Now, our main theorem is
\begin{thm}
\label{thm.main}
We have $e_\textnormal{Jel} = e_\textnormal{per} = e_\textnormal{UEG}$.
\end{thm}

\noindent
The analogous result in dimensions $d\geq 3$ was proven by Cotar and Petrache \cite{cotar.petrache.2019}
using methods of optimal transport and later in dimension $d=3$ by Lewin, Lieb and Seiringer \cite{lewin.lieb.seiringer.19}
using a ``floating crystal'' trial state, which our method builds on.
Cotar and Petrache \cite[Remark 1.7]{cotar.petrache.2019} note that the case of $d=2$ 
is an open problem. Our findings here thus solves this open problem.

One inequality is the following argument.
Let $\P$ be any $N$-particle probability measure with $\rho_\P = \mathbbm{1}_{\Omega_N}$. Then,
\begin{multline*}
- \int \sum_{j < k} \log|x_j - x_k| \ud \P(x_1, \ldots, x_N) + \frac{1}{2} \iint_{\Omega_N \times \Omega_N} \log |x-y| \ud x \ud y
	\\
	 = \int \mathcal{E}_{\textnormal{Jel}}(\Omega_N, x_1, \ldots, x_N) \ud \P(x_1, \ldots, x_N)
	 \geq \min \mathcal{E}_{\textnormal{Jel}}(\Omega_N, x_1, \ldots, x_N).
\end{multline*}
Optimising over $\P$ and taking the thermodynamical limit we thus get
$e_\textnormal{EUG} \geq e_\textnormal{Jel}$.
In order to get the inequality $e_\text{UEG} \leq e_\text{per} \leq e_\text{Jel}$ 
we will superficially introduce a crystal structure to the Jellium configuration.
This is similar to (and inspired by) the floating crystal argument from \cite{lewin.lieb.seiringer.19}. 
We give the proof in sections \ref{sect.upper.bound.ueg} and \ref{sect.upper.bound.per}.

\section{Lattice Configurations}%\label{sect.lattice.configs}
We now consider the Jellium energies of periodic configurations,
when the electrons are positioned on a lattice.
We will consider these configurations in any dimension $d$ and
for general Riesz interactions. 
These we first define.
For $s\in \R$ the Riesz potential $V_s$ on $\R^d$ is given by
\[
	V_s(x) = \begin{cases}
	|x|^{-s} & \text{if } s > 0, \\
	-\log |x| & \text{if } s = 0, \\
	-|x|^{-s} & \text{if } s < 0.
	\end{cases}
\]
Then $V_{d-2}$ is the Coulomb potential in $d$ dimensions. 
With this we may define for $s < d$ the Jellium energy in $d$ dimensions with potential $V_s$,
\[
	\mathcal{E}_{\textnormal{Jel}, d, s}(\Omega_N, x_1,\ldots,x_N)
		= \sum_{j < k} V_s(x_j - x_k) - \sum_{j=1}^N \int_{\Omega_N} V_s(x_j - y) \ud y
			+D_{d,s}\left(\mathbbm{1}_{\Omega_N}\right),
\]
where $D_{d,s}(f,g) = \frac{1}{2} \iint_{\R^d \times \R^d} f(x)g(y) V_s(x-y) \ud x \ud y$.
Define for a lattice $\mathcal{L}\subset \R^d$ with Wigner-Seitz unit cell $Q$ with $|Q| = 1$ 
and $s$ satisfying $d-4 < s < d$ the energy
\[
	e_{\textnormal{Jel}, s}^\mathcal{L} 
		= \lim_{\Omega_N\nearrow \R^2} \frac{\mathcal{E}_{\textnormal{Jel}, d, s}(\Omega_N, x_1,\ldots,x_N)}{|\Omega_N|}
\]
as the thermodynamic limit of Jellium, when the electrons are placed on the lattice. 
(The existence of this thermodynamic limit follows from the proof of the theorem below.)
Here $\Omega_N = \bigcup_{i=1}^N \left(Q+x_i\right)$.
Define for $\textnormal{Re}(s) > d$ the Epstein (or lattice) $\zeta$-function
\[
	\zeta_{\mathcal{L}}(s) = \frac{1}{2} \sum_{x\in \mathcal{L}\setminus 0} \frac{1}{|x|^s}.
\]
This function has a meromorphic continuation to all of $\C$ with a simple pole at $s=d$, see \cite{borwein.borwein.straub.13}.
These more complicated $\zeta$-functions can oftentimes be expressed in terms of simpler functions, see \cite{zucker.robertson.75}.
We prove the following.
\begin{thm}
\label{thm.lattice.jellium}
Let $s$ satisfy $d-4 < s < d$ and 
let $\mathcal{L}\subset \R^d$ be a lattice with Wigner-Seitz unit cell $Q$, $|Q| = 1$.
Then the Jellium energy of the lattice configuration is
\[
	e_{\textnormal{Jel}, s}^\mathcal{L}
		= 
			\begin{cases}
				\zeta_\mathcal{L}(s) & \textnormal{if } s > 0,
				\\
				\zeta_\mathcal{L}'(0) & \textnormal{if } s = 0,
				\\
				 -\zeta_\mathcal{L}(s) & \textnormal{if } s < 0.
			\end{cases}
\]
\end{thm}

\noindent
Many similar results exist in the literature. In \cite{borwein.borwein.straub.13,borwein.borwein.shail.1989}
a similar result is shown for a slightly different energy functional in the case $d-2 < s < d$
via analytic continuation of the Epstein $\zeta$-function.
Extending upon these ideas a partial result for the Jellium energy is shown in \cite[Appendix B]{lewin.lieb.2015}.
Other formulations also exist, 
see for instance \cite{kuijlaars.saff.1998,brauchart.hardin.saff.2012,betermin.sandier.18}
for the case of logarithmically interacting points on the unit 2-sphere,
and \cite{sandier.serfaty.12,betermin.sandier.18}
for the case of logarithmically interacting points on the plane.
As we have not found a complete proof of \Cref{thm.lattice.jellium} in the literature, 
we give a straightforward proof in \cref{sect.lattice}.
As an application of the theorem, we compute the energy density of Jellium in the triangular lattice (in 2 dimensions). 
\begin{ex}
The triangular lattice is given by $\mathcal{L} = c(1,0)\Z \oplus c\left(\frac{1}{2}, \frac{\sqrt{3}}{2}\right) \Z$, where the constant $c$ is such that 
$|Q|=1$, i.e. $c^2 = \frac{2}{\sqrt{3}}$.
Thus by \cite{zucker.robertson.75} we have
\[
	\zeta_\mathcal{L}(s) = \frac{1}{2}\sum_{x\in \mathcal{L}\setminus0} \frac{1}{|x|^s}
		= \frac{1}{2c^s} \sum_{(n,m) \in \Z^2\setminus (0,0)} \frac{1}{(n^2 + mn + m^2)^{s/2}}
		= 3 c^{-s} \zeta\left(\frac{s}{2}\right)L_{3}\left(\frac{s}{2}\right),
\]
where $\zeta$ is the Riemann zeta-function and 
$L_{3}(s) = L(s,\chi)$ is the Dirichlet $L$-series for the nontrivial character mod 3, i.e.
\[
	L_{3}(s) = \sum_{n=1}^\infty \frac{\chi(n)}{n^s} = 1 - 2^{-s} + 4^{-s} - 5^{-s} + \ldots 
	= 3^{-s}\left(\zeta\left(s, 1/3\right) - \zeta\left(s, 2/3\right)\right)
\]
where $\zeta(s, a)$ is the Hurwitz $\zeta$-function. The values of these functions and their derivatives can be found in \cite{dlmf}.
We conclude that 
$e_{\textnormal{Jel}, s=0}^\mathcal{L} = \zeta_\mathcal{L}'(0) =  \frac{1}{8} \log\left(\frac{48 \pi}{\Gamma\left(1/6\right)^6}\right) \simeq -0.66056$. 
In comparison, the best-known lower bound \cite{sari.merlini.76,lieb.narnhofer.75} is 
$e_{\textnormal{Jel}, s=0} \geq -\left(\frac{3}{8} + \frac{1}{4} \log \pi\right) \simeq -0.66118$.
The triangular lattice, which is what we expect to be the ground state, is remarkably close to this lower bound. 
\end{ex}	
\begin{ex}
In dimension $d=1$, there is only one lattice, namely $\Z$. Thus
$e_{\text{Jel}, s=-1}^{\Z} = -\zeta(-1) = \frac{1}{12}$.
It is in fact known that Jellium is crystalised in one dimension \cite{kunz.74,baxter.63,choquard.75}.
\end{ex}

\section{Relation to the Renormalized Energy}\label{sect.renorm.energy}
We now relate the Jellium energy to the renormalized energy defined in \cite{sandier.serfaty.12}.
This renormalized energy has appeared before in the study of Ginzburg-Landau theory 
and of Coulomb gases, see 
\cite{sandier.serfaty.12,borodin.serfaty.13,sandier.serfaty.15,petrache.serfaty.2017,rougerie.sefaty.16,leble.serfaty.17,cotar.petrache.2019}.
We recall the definition here.
\begin{defn}
Let $j$ be a vector-field on $\R^2$. Let $m > 0$. We say that $j\in \mathcal{A}_m$ if
\begin{equation}
\label{eqn.cond.j}
	\curl j = 2\pi (\nu - m),
	\qquad 
	\div j = 0,
	\qquad
	\sup_{R > 1} \frac{\nu(B_R)}{|B_R|} < \infty,
\end{equation}
where $\nu = \sum_{p \in \Lambda} \delta_p$ for a discrete set $\Lambda \subset \R^2$.
\end{defn}	

\noindent
Then, for any function $\chi$ we define 
\[
	W(j, \chi) = \lim_{\eta \to 0}\left(\frac{1}{2}\int_{\R^2 \setminus \bigcup_{p\in \Lambda}B(p,\eta)} \chi |j|^2	+ \pi \log \eta \sum_{p\in \Lambda} \chi(p)\right).
\]
The renormalized energy is then defined as
\begin{defn}
The renormalized energy of $j$ is
\[
	W(j) := \limsup_{R\to \infty} \frac{W(j, \chi_R)}{|B_R|},
\]
where $\chi_R$ denotes any cutoff-functions satisfying
\[
	|\nabla \chi_R| \leq C, 
	\qquad
	\supp(\chi_R) \subset B_R,
	\qquad
	\chi_R(x) = 1 \textnormal{ on } B_{R-1}.
\]
\end{defn}

\noindent
We recall a few properties of $W$ from \cite{sandier.serfaty.12}.
\begin{itemize}
\item
The renormalized energy $W(j)$ 
does not depend on the choice of cut-off functions $\chi_R$. 
\item If $j \in \mathcal{A}_m$ then $j' = \frac{1}{\sqrt{m}} j(\cdot / \sqrt{m}) \in \mathcal{A}_1$ 
and $W(j) = m\left(W(j') - \frac{1}{4} \log m\right)$.
In particular 
$\min_{\mathcal{A}_m} W = m\left(\min_{\mathcal{A}_1} W - \frac{1}{4}\log m\right)$.
\item
If $j\in \mathcal{A}_m$ with $W(j) < \infty$ then 
$\lim_{R \to \infty} \frac{\nu(B_R)}{|B_R|} = m$.
\item
$\min_{\mathcal{A}_1} W$ is the limit of a sequence of periodic configurations with period $n\to \infty$.
\end{itemize}

\noindent
For periodic $\Lambda$ we have the following result.
\begin{prop}[{\cite[Proposition 3.1]{sandier.serfaty.12}}]\label{prop.w.periodic}
Suppose $\Lambda$ is periodic with respect to some lattice $\mathcal{L}$
and denote the points of $\Lambda$ in the torus $\mathbb{T} = \R^2 / \mathcal{L}$ by $\{x_1, \ldots, x_n\}$.
Define $H_{\{x_i\}}$ and $j_{\{x_i\}}$ on $\mathbb{T}$ by 
\[
	-\Delta H_{\{x_i\}} = 2\pi \left(\sum_{i=1}^n \delta_{x_i} - \frac{n}{|\mathbb{T}|}\right),
	\qquad
	j_{\{x_i\}} := - \nabla^\perp H_{\{x_i\}}.
\] 
Then $W(j) \geq W(j_{\{x_i\}})$ for any $j$ satisfying \Cref{eqn.cond.j}.
\end{prop}

\noindent
Note that $H_{\{x_i\}}$ is defined uniquely up to a constant, and thus $j_{\{x_i\}}$ is well-defined.
Now, the relation of this renormalized energy to the Jellium energy is the following.
\begin{cor}
\label{cor.renorm.energy}
The renormalized energy is given by
\[
	\min_{\mathcal{A}_1} W = 2 \pi e_\textnormal{Jel}.
\]
\end{cor}

\noindent
Cotar and Petrache \cite{cotar.petrache.2019} 
have shown a similar result for more general Riesz interactions, but not including the $d=2, s=0$ case, which is the one considered here,
see \cref{remark.renorm.s.d}.
\begin{proof}
Since $\min_{\mathcal{A}_1} W$ is the limit of periodic configurations \cite[Theorem 1]{sandier.serfaty.12} 
and any such periodic configuration clearly has energy at least  $\min_{\mathcal{A}_1} W$ we have
\[
	\min_{\mathcal{A}_1} W = \lim_{\ell \to \infty} \min_{\substack{j \in \mathcal{A}_1 \\ j \textnormal{ is } \ell\textnormal{-periodic}}} W(j)
\]
Now, suppose $j$ is $\ell$-periodic.
Denote the points of $\Lambda$ in the torus $\R^2/\ell \Z^2$ by $\{x_1, \ldots, x_n\}$.
Then $W(j) \geq W \left(j_{\{x_i\}}\right)$ by \Cref{prop.w.periodic}.
Now, by \cite[Lemma 2.7]{borodin.serfaty.13}
\[
	W \left(j_{\{x_i\}}\right) = \frac{2\pi}{n} \sum_{i < j} G_\ell(x_i - x_j) + \pi \lim_{x \to 0}\left( G_\ell(x) + \log |x|\right)
		= \frac{2\pi}{n} \mathcal{E}_{\text{per}, \ell}(x_1,\ldots,x_n).
\]
Hence by \Cref{thm.main}
\[
	\min_{\mathcal{A}_1} W = 2\pi \lim_{n\to \infty} \min_{x_1,\ldots,x_n} \frac{\mathcal{E}_{\text{per}, \ell}(x_1,\ldots,x_n)}{n}
		= 2\pi e_\text{per} = 2 \pi e_\text{Jel}.
	\qedhere
\]
\end{proof}

\begin{remark}
With this, the known bounds on $e_\text{Jel}$ carry over. The upper bound of 
\[
	\min_{\mathcal{A}_1} W = 2\pi e_\text{Jel} \leq 2\pi e_\text{Jel}^{\mathcal{L}} \simeq -4.1504,
\]
where $\mathcal{L}$ is the triangular lattice
was previously known \cite{sandier.serfaty.12,betermin.sandier.18}.
The lower bound $e_\text{Jel} \geq - \left(\frac{3}{8} + \frac{1}{4}\log\pi\right)$ \cite{sari.merlini.76,lieb.narnhofer.75} 
gives the bound
\[
	\min_{\mathcal{A}_1} W = 2 \pi e_\text{Jel} \geq -\pi \left(\frac{3}{4} + \frac{1}{2} \log \pi\right) \simeq -4.1543.
\]
This is an improvement on previous known lower bounds.
The previous known lower bound by Steinerberger \cite{steinerberger.20}, 
translated to this setting using the results of \cite{betermin.sandier.18}, 
is 
$\min_{\mathcal{A}_1} W \geq -\frac{\pi}{2}(1 + \gamma + \log \pi)\simeq -4.2756$,
where $\gamma \simeq 0.577$ is the Euler--Mascheroni constant.
\end{remark}

\begin{remark}\label{remark.log.s2}
The Renormalized energy has also been used by Bétermin and Sandier \cite{betermin.sandier.18} in the problem of 
optimal point-configurations on the sphere $\mathbb{S}^2$ with a logarithmic energy functional,
i.e. points $x_1,\ldots,x_n \in \mathbb{S}^2$ minimising 
$E(x_1,\ldots,x_n) = - \sum_{i\ne j} \log |x_i - x_j|$.
The problem of minimising the logarithmic energy of points on the sphere 
has received much study, see \cite{betermin.sandier.18,brauchart.hardin.saff.2012,kuijlaars.saff.1998,dubickas.1996,steinerberger.20},
and is linked to Smale's 7th problem, see \cite{beltran.2012} for a review.

Let $\mathcal{E}_\text{log}(n) = \min E(x_1,\ldots,x_n)$
denote the minimal energy. 
Bétermin and Sandier \cite{betermin.sandier.18}
showed that there exists a constant $c_{\text{log}} = \frac{1}{\pi}\min_{\mathcal{A}_1} W + \frac{\log 4\pi}{2}$ such that
\[
	\mathcal{E}_\text{log}(n) = \left(\frac{1}{2} - \log 2\right) n^2 - \frac{1}{2}n \log n + c_{\text{log}} n + o(n)
\]
as $n\to \infty$.
Written in terms of this constant $c_{\text{log}}$ the improved lower bound says
\[
	c_{\text{log}} 
	= 2e_\text{Jel} + \frac{\log 4\pi}{2}
		%= \frac{1}{\pi}\min_{\mathcal{A}_1} W + \frac{\log 4\pi}{2} 
		%= 2e_\text{Jel} + \frac{\log 4\pi}{2} 
		\geq \log 2 - \frac{3}{4} \simeq -0.0569.
\]
The previous known lower bound by Steinerberger \cite{steinerberger.20} is $c_\text{log} \geq \frac{\log 4 - 1 - \gamma}{2} \simeq -0.0954$.
\end{remark}

\begin{remark}\label{remark.renorm.s.d}
The Renormalized energy has also been defined for general Riesz potentials in \cite{petrache.serfaty.2017},
and Jellium and periodic Jellium in \cite{lewin.lieb.seiringer.19}.
Let $\mathcal{W}$ be as defined in \cite[Definition 1.3]{petrache.serfaty.2017} (this differs from the $W$ considered above 
by a factor of 2 in the case $d=2,s=0$) and $e_\text{per}(d,s)$ the periodic Jellium energy 
for Riesz potential with parameter $s$ and in dimension $d$ as defined in \cite{lewin.lieb.seiringer.19}.
Exactly the same proof as above shows that
$\min_{\mathcal{A}_1} \mathcal{W} = 2c_{d,s} e_\text{per}(d,s)$,
for $\max(0,d-2) \leq s < d$, where 
\[
	c_{d,s} = \begin{cases}
	2s \frac{2\pi^{d/2} \Gamma\left(\frac{s+2-d}{2}\right)}{\Gamma\left(\frac{s+2}{2}\right)}
	& \text{if } \max(0, d-2) < s < d,
	\\
	(d-2) \frac{2\pi^{d/2}}{\Gamma(d/2)}
	& \text{if } s=d-2 > 0,
	\\
	2\pi
	& \text{if } s=0, d=1,2.
	\end{cases}
\]
In \cite{lewin.lieb.seiringer.19,petrache.serfaty.2017,leble.serfaty.17,cotar.petrache.2019} and \Cref{thm.main} 
it is proved that $e_\text{Jel}(d,s) = e_\text{per}(d,s)$
for the relevant $d,s$. Thus we have
$\min_{\mathcal{A}_1} \mathcal{W} = 2c_{d,s} e_\text{Jel}(d,s)$.
This result was previously shown in \cite{cotar.petrache.2019} only not including the case $d=2, s=0$.
\end{remark}

\noindent
We now turn to the proofs of Theorems \ref{thm.main} and \ref{thm.lattice.jellium}.

\section{Upper Bound for the Uniform Electron Gas Energy}
\label{sect.upper.bound.ueg}
We first show that 
\[
	e_\text{UEG} \leq \frac{\mathcal{E}_{\text{per}, \ell}(x_1, \ldots, x_n)}{n}.
\]
The proof is very similar to the proof of the same result in dimension $d=3$ presented
in \cite{lewin.lieb.seiringer.19}. 
The main difference is the different choice of trial state $\P$.
As discussed above, we cannot just use the trial state considered in \cite{lewin.lieb.seiringer.19},
since in 2 dimensions we do not have the more relaxed formulation
of the thermodynamic limit for the UEG, which is used in \cite{lewin.lieb.seiringer.19}.
It is not enough that the density $\rho_\P$ is 1 in the bulk of 
$\Omega_N$, 0 outside and bounded close to the boundary.
Secondly, due to the long range behaviour of the logarithm some 
error bounds are slightly more complicated in 2 dimensions.

The construction of the trial state $\P$ is similar to that of \cite{lewin.lieb.seiringer.19}.
We consider a floating crystal immersed in a thin fluid layer.
We want our density $\rho_\P$ to be a characteristic function. 
In particular, we do not allow for the fluid to ever be where the crystal might also be 
(under different translations), i.e. we make a hole in the fluid which is larger than the crystal. 
Additionally, the fluid layer is not chosen to have constant density $1$,
but instead have some density profile $\beta$, which we describe below.

Consider any arrangement of $n$ points $x_1, \ldots, x_n$ in the cube $C_\ell$ of side length $\ell = n^{1/2}$ (centered at 0).
Adding a background shifted by the center of mass $\tau = \frac{1}{n}\sum_{j=1}^n x_j$ we get an arrangement with
no dipole moment: 
$\int y \left(\sum_{j=1}^n \delta_{x_j}(y) - \mathbbm{1}_{C_\ell + \tau}(y)\right) \ud y = 0$.
We copy this arrangement periodically in the larger cube
\[
	\Omega_N = \bigcup_{\substack{k\in \Z^2 \\ |k_1|, |k_2| \leq K}} \left(C_\ell + \ell k\right)
\]
of volume $N = \ell^2(2K + 1)^2$. 
The electrons are localised at the points $x_{j} = x_{j_0} + \ell k$. 
We now first show that 
\[
	e_\textnormal{UEG} 
		\leq \liminf_{N\to \infty} \frac{\mathcal{E}_\text{Jel}\left(\Omega_N + \tau, x_1, \ldots, x_N\right)}{N}.
\]

\noindent
Let $C$ be a square with 
$C \supset \Omega_N + 5 C_\ell$ and 
$|C\setminus \Omega_N| = M = O\left(N^{1/2}\right)$ an integer.
Define
$F = \Omega_N + 3 C_\ell + \tau$ (the hole in the fluid) and 
$\beta = \mathbbm{1}_C + \mathbbm{1}_F * \frac{\mathbbm{1}_{C_\ell}}{\ell^2} - \frac{1}{n}\sum_{j=1}^n \mathbbm{1}_{\Omega_N + x_j}$.
Then define
\[
	\P = \frac{1}{\ell^2}\int_{C_\ell} \bigotimes_{\substack{j=1,\ldots,n \\ k\in \Z^2 \\ |k_1|, |k_2| \leq K}}
		\delta_{x_j + \ell k + a} \otimes
		\left(\frac{\beta - \mathbbm{1}_{F+a}}{M}\right)^{\otimes M} \ud a.
\]
We choose $F$ such that the fluid and the crystal never overlap. 
The crystal, when shifted around, gives some density in the region $\Omega_N + C_\ell$. We thus want that any shift $F + a$ of $F$
doesn't overlap with this. Since $\tau$ is in general just some vector $\tau \in C_\ell$, this leads to our definition of 
$F = \Omega_N + 3C_\ell + \tau$. We need $\beta \geq \mathbbm{1}_{F+a}$ for any shift $a$, and so this leads to $\beta \geq \mathbbm{1}_C$,
with $C \supset \Omega_N + 5 C_\ell$. This is true by construction.
We compute that
\[
	\rho_\P = \beta + \frac{1}{n}\sum_{j=1}^n \mathbbm{1}_{\Omega_N + x_j} - \mathbbm{1}_F * \frac{\mathbbm{1}_{C_\ell}}{\ell^2} = \mathbbm{1}_C.
\]
Thus, $\rho_\P$ is a characteristic function as desired. 

In principle, we could have taken $\beta$ any function with
$\mathbbm{1}_C \leq \beta \leq \mathbbm{1}_C + \mathbbm{1}_F * \frac{\mathbbm{1}_{C_\ell}}{\ell^2} - \frac{1}{n}\sum_{j=1}^n \mathbbm{1}_{\Omega_N + x_j}$,
since then the density $\rho_\P$ would satisfy $\rho_\P \leq 1$ 
and it is trivial to extend the thermodynamic limit \cite[Theorem 2.6]{lewin.lieb.seiringer.18}
also to this case. 
In the computations below we will use the description of $\rho_\P$ in terms of $\beta$,
as this will make the computations slightly nicer.

To compute the energy $\mathcal{E}_\text{Ind}(\P)$ we first introduce the notation
\[
	D(\mu, \nu) = \frac{1}{2} \iint -\log|x-y| \ud \mu(x) \ud \nu(y),
	\qquad
	D(\mu) = D(\mu, \mu)
\]
for two (signed) measures $\mu, \nu$. This is the Coulomb interaction energy
between charge distributions $\mu$ and $\nu$.
Mostly we will use this in the case where the measures are given by functions.
The analogous object in dimensions $d\geq 3$ has $D(f) \geq 0$ for any function.
This is however not true in dimension $d=2$. In general it is only true for 
functions with zero mean.
\begin{prop}
\label{prop.d(f).geq.0}
Suppose $f$ has $\int f \ud x =0$. Then $D(f) \geq 0$.
\end{prop}
\begin{proof}
By density we may assume that $f\in \mathcal{S}$, i.e. that $f$ is rapidly decreasing. Define $f^\#(x) = f(-x)$. 
Note that $\widehat{f^\#} = \overline{\hat{f}}$.
First we show that $\frac{\widehat{f^\#} \hat f}{p^2} = \frac{|\hat f(p)|}{p^2}$
is the Fourier transform of some function. 
Define $g:= \frac{-1}{2\pi}\log * f^\# * f$.
Then by \cite[Theorem 6.21]{analysis} we have that $g\in L^1_\text{loc}$ and $-\Delta g = f^\# * f$ in $\mathcal{D}'$.
Since $\log \in \mathcal{S}'$ we have that $g\in \mathcal{S}'$ and so
$p^2 \hat g = \widehat{f^\# * f}$ in $\mathcal{S}'$. Hence $\hat g(p) = 2\pi \frac{|\hat f(p)|^2}{p^2}$ as functions.
By the assumption $\int f\ud x = 0$ we have that the right-hand-side actually stays bounded (and smooth) as $p\to 0$. We conclude that 
$\hat g\in \mathcal{S}$ and so $g\in \mathcal{S}$ has a Fourier transform as a function.

Now, with $\ip{\cdot}{\cdot}$ denoting application of a distribution we have
\[
	\ip{-\log * f}{f} = \ip{\widehat{-\log}}{\widehat{f^\# *f}}
		= 2\pi \ip{p^2\cdot \widehat{-\log}}{\frac{|\hat f(p)|^2}{p^2}}
%		= 2\pi \ip{\frac{|\hat f(p)|^2}{p^2}}{1}
		=2\pi \int \frac{|\hat f(p)|^2}{p^2}\ud p \geq 0,
\]
since by \cite[Theorem 6.20]{analysis} we have $-\Delta(-\log|\cdot|) = 2\pi \delta$ in $\mathcal{D}'$.
\end{proof}
\begin{remark}
In dimension $d=1$ an analogous statement also holds. 
\end{remark}

\noindent
Now, we may calculate the energy (with $x_j$, for $j > n$ denoting the points $x_j + \ell k$ for $k\ne 0$)
\begin{eqnarray*}
\mathcal{E}_\text{Ind}(\P)
	& = &
		\sum_{1\leq j<k\leq N} -\log|x_j - x_k|
		+ \sum_{j=1}^N \frac{1}{\ell^2} \int_{C_\ell} \int -\log|x_j + a - y|\left( \beta(y) - \mathbbm{1}_{F + a}(y) \right)\ud y \ud a
	\\
	& &
		+ \left(1 - \frac{1}{M}\right) \frac{1}{\ell^2} \int_{C_\ell} D\left(\beta - \mathbbm{1}_{F+a}\right) \ud a 
		- D(\rho_\P)
	\\
	& = &
		\sum_{j<k} -\log|x_j - x_k|
		+\frac{1}{n} \sum_{j=1}^n 2D\left(\mathbbm{1}_{\Omega_N + x_j}, \beta\right) 
		- 2D\left(\sum_{j=1}^N \delta_{x_j}, \mathbbm{1}_F\right)
	\\
	& &
		+ D(\beta) 
		+ D \left(\mathbbm{1}_F\right)
		-2 D\left(\beta, \mathbbm{1}_{F} * \frac{\mathbbm{1}_{C_\ell}}{\ell^2}\right)
		-D\left(\rho_\P\right)
		-\frac{1}{M\ell^2}\int_{C_\ell} D \left(\beta - \mathbbm{1}_{F+1}\right) \ud a
	\\
	& = &
		\mathcal{E}_\text{Jel}\left(\Omega_N + \tau, x_1, \ldots, x_N\right)
		+ 2D\left(\sum_{j=1}^N \delta_{x_j}, \mathbbm{1}_{\Omega_N + \tau} - \mathbbm{1}_F\right)
		- D\left(\mathbbm{1}_{\Omega_N + \tau}\right)
	\\
	& & 
		+ 2D\left(\beta, \frac{1}{n}\sum_{j=1}^n \mathbbm{1}_{\Omega_N + x_j} - \mathbbm{1}_F * \frac{\mathbbm{1}_{C_\ell}}{\ell^2}\right)
		+ D(\beta)
		+ D\left(\mathbbm{1}_{F}\right)
		- D\left(\rho_P\right)
	\\
	& &
		- \frac{1}{M\ell^2} \int_{C_\ell} D \left(\beta - \mathbbm{1}_{F + a}\right) \ud a
\end{eqnarray*}
First, we claim that
\[
	- \frac{1}{M\ell^2} \int_{C_\ell} D \left(\beta - \mathbbm{1}_{F + a}\right) \ud a \leq o(N).
\]
\begin{remark}
In dimensions $d\geq 3$ the analogous term is $\leq 0$ since $D(f)\geq 0$ for any function $f$.
\end{remark}

\noindent
For any $a$ denote by $A = \supp \left(\beta - \mathbbm{1}_{F + a}\right) $. 
Then we have $|A| = O\left(N^{1/2}\right)$ and $\diam A = O(L)=O\left(N^{1/2}\right)$.
Thus
\begin{multline*}
D \left(\beta - \mathbbm{1}_{F + a}\right)
%	& 
	= \frac{1}{2}\iint_{A\times A} -\log|x-y| \left(\beta - \mathbbm{1}_{F + a}\right)(x) \left(\beta - \mathbbm{1}_{F + a}\right)(y) \ud x \ud y
	\\
%	& 
	\geq C\iint_{A\times A} -\log\diam A \ud x \ud y
%	\\
%	& 
%	= - \frac{1}{2}|A|^2 \log\diam A
%	\\
%	& 
	= O(N \log N).
\end{multline*}
Hence
\[
	- \frac{1}{M\ell^2} \int_{C_\ell} D \left(\beta - \mathbbm{1}_{F + a}\right) \ud a 
		\leq O(N^{1/2}\log N) 
		= o(N).
\]
We are thus left with the error term
\begin{multline*}
	2D\left(\sum_{j=1}^N \delta_{x_j}, \mathbbm{1}_{\Omega_N + \tau} - \mathbbm{1}_F\right)
		+ 2D\left(\beta, \frac{1}{n}\sum_{j=1}^n \mathbbm{1}_{\Omega_N + x_j} - \mathbbm{1}_F * \frac{\mathbbm{1}_{C_\ell}}{\ell^2}\right)
	\\
		- D\left(\mathbbm{1}_{\Omega_N + \tau}\right)
		+ D(\beta)
		+ D\left(\mathbbm{1}_{F}\right)
		- D\left(\rho_P\right)
\end{multline*}
Plugging in the value of $\rho_\P$ we may calculate this term as
\begin{multline}\label{eqn.error.bound.ueg}
	-2 D\left(\sum_{j=1}^N \delta_{x_j} - \mathbbm{1}_{\Omega_N + \tau}, f\right)
	+ D\left(\mathbbm{1}_F - \mathbbm{1}_F * \frac{\mathbbm{1}_{C_\ell}}{\ell^2}, 
		g\right)	
	+ D\left(\frac{1}{n}\sum_{j=1}^n \mathbbm{1}_{\Omega_N + x_j} - \mathbbm{1}_{\Omega_N + \tau},
		g\right),
\end{multline}
where 
\[
	f = \mathbbm{1}_F - \mathbbm{1}_{\Omega_N + \tau},
	\qquad 
	g = \mathbbm{1}_F + \mathbbm{1}_F * \frac{\mathbbm{1}_{C_\ell}}{\ell^2} - \mathbbm{1}_{\Omega_N + \tau}	- \frac{1}{n}\sum_{j=1}^n \mathbbm{1}_{\Omega_N + x_j}.
\]
We claim that \Cref{eqn.error.bound.ueg} is $O(N^{1/2}\log N)$ and thus vanishes in the desired limit.
This will follow from appropriate Taylor expansions of $-\log|\cdot|$ and the following two propositions.
\begin{prop}\label{prop.n.1/2.logn}
Let $A$ be a square of size $|A| = O(N)$. Let $\mu$ be a measure satisfying
$\mu\left(B_R\right) = O\left(R^2\right)$ as $R\to \infty$.
Let $B$ be the boundary region of $A$, meaning 
$B = \{x \in \R^2: \ud(x, \partial A) \leq \ell\}$ for some fixed $\ell > 0$. Then
\[
	\int_A \int_{B\cap \{|x-y| > 1\}} \frac{1}{|x-y|^2} \ud y \ud \mu(x) = O\left(N^{1/2}\log N\right).
\]
\end{prop}

\noindent
One should think that $\mu$ is either Lebesgue measure or a sum of appropriately distributed $\delta$-measures.
To show this, note that 
for any fixed $y\in B$ we can bound the $x$-integral by the integral over a ball of radius $L=O\left(N^{1/2}\right)$ centered at $y$
(removing the ball of radius 1).
Thus
\[
	\int_A \int_{B \cap \{|x-y| > 1\}} \frac{1}{|x-y|^2} \ud y \ud \mu(x)
		\leq \int_B \int_{B_L \setminus B_1} \frac{1}{|z|^2} \ud \mu(z) \ud y
		\leq C \int_B \log L \ud y = O\left(N^{1/2} \log N\right).
\]
\begin{prop}\label{prop.d.convolve.dipole.=0}
Let $A, B$ be as in \Cref{prop.n.1/2.logn} and $\mu$ be a probability measure supported in $B_\ell$. 
Denote by $\tau$ the first moment of $\mu$, i.e. $\tau = \int a \ud \mu(a)$. 
Let $b$ be a function supported in $B$, which is bounded uniformly in $N$.
Then
\[
	D\left(\mathbbm{1}_{A + \tau} - \mathbbm{1}_A * \mu, b\right)
	= O\left(N^{1/2}\log N\right).
\]
\end{prop}

\begin{remark}
In dimension $d\geq 1$ we similarly have
\begin{align*}
	\int_A \int_{B \cap \{|x-y| > 1\}} \frac{1}{|x-y|^d} \ud y \ud \mu(x)
	& = O\left(N^{\frac{d-1}{d}}\log N\right),
\\
	D\left(\mathbbm{1}_{A + \tau} - \mathbbm{1}_A * \mu, b\right)
	& = O\left(N^{\frac{d-1}{d}}\log N\right).
\end{align*}
Thus our argument also works in higher dimensions.
\end{remark}

\noindent 
We postpone the proof of \Cref{prop.d.convolve.dipole.=0} to the appendix.
\Cref{prop.d.convolve.dipole.=0} immediately gives that the second and third terms of \Cref{eqn.error.bound.ueg} are $O(N^{1/2} \log N)$.
For the first term we use that by Taylor expansion
\begin{align*}
	-\log|z+a| 
	& = -\log|z| - \frac{z \cdot a}{|z|^2}
		+ \int_0^1 (1-t) \left[\frac{|a|^2}{|z + ta|^2} - \frac{2 ((z + ta)\cdot a)^2}{|z + ta|^4}\right] \ud t
	\\
	& = -\log|z| - \frac{z \cdot a}{|z|^2} + O\left(\frac{1}{|z|^2}\right)
\end{align*}
for $a$ bounded and $|z|$ bounded from below.
Thus for the term
\[ 
	2D\left(\sum_{j=1}^N \delta_{x_j} - \mathbbm{1}_{\Omega_N + \tau}, f\right)
	=
			\iint f(x) (-\log|x-y|) \left(\sum_{j=1}^N \delta_{x_j} - \mathbbm{1}_{\Omega_N + \tau}\right)(y) \ud y \ud x
\]
we have
\begin{align*}
	& = 
			\sum_{\substack{j=1,\ldots,n \\ k \in \Z^2 \\ |k_1|, |k_2| \leq K}} 
				\iint f(x)(-\log|x-y|) \left[\delta_{x_j + \ell k} - \frac{1}{\ell^2} \mathbbm{1}_{C_\ell + \tau + \ell k}\right](y) \ud y \ud x
		\\
	& = 
			\sum_{j, k} \frac{1}{\ell^2}
				\int_{F\setminus (\Omega_N + \tau)} \int_{C_{\ell}} -\log|x-\ell k - x_j| + \log|x-\ell k - \tau + a| \ud a \ud x
		\\
\intertext{Hence, since $|x - k\ell|$ is bounded from below on $F\setminus(\Omega_N + \tau)$ we have}
	& = 
			\int_{F\setminus (\Omega_N + \tau)} 
				\underbrace{\sum_{j, k} - \frac{(x - \ell k - \tau) \cdot (\tau - x_j)}{|x - \ell k - \tau|^2}}_{= 0 \text{ by } \sum x_j = \sum \tau} 
				\ud x 
			+ \sum_{j,k} \int_{F\setminus (\Omega_N + \tau)} O\left(\frac{1}{|x-\ell k|^2}\right) \ud x
		\\
	& \quad
			+ \sum_{j, k} \frac{1}{\ell^2} \int_{F\setminus (\Omega_N + \tau)} 
				\underbrace{\int_{C_\ell} \frac{(x - \ell k - \tau) \cdot a}{|x- \ell k - \tau|^2} \ud a}_{= 0 \text{ by } \int a \ud a = 0} 
				\ud x 
			+ \sum_{j,k} \int_{F\setminus (\Omega_N + \tau)} O\left(\frac{1}{|x-\ell k|^2}\right) \ud x
		\\
	& = 
			O(N^{1/2}\log N)
\end{align*}
by \Cref{prop.n.1/2.logn}.
This gives the bound
$\mathcal{E}_\text{Ind}(\rho_\P) \leq \mathcal{E}_\text{Ind}(\P) \leq \mathcal{E}_\text{Jel}\left(\Omega_N + \tau, x_1, \ldots, x_N\right) + o(N)$.
Hence, by taking the thermodynamic limit we get the desired.

We now show that
\[
	\lim_{N\to \infty} \frac{\mathcal{E}_\text{Jel}(\Omega_N + \tau, x_1, \ldots, x_N)}{N} = \frac{\mathcal{E}_{\text{per}, \ell}(x_1, \ldots, x_n)}{n}.
\]
This argument is more or less the same as in \cite{lewin.lieb.seiringer.19}.
There are slight differences in the case $d=2$ compared to the case $d\geq 3$, 
which is why we present the argument here.
We really only need the bound $\leq$, and this is what we now show.

Note that the inter-particle distance is bounded uniformly from below (since there are only finitely many particles in the ``unit cell'' $C_\ell$).
Hence, by replacing the point charges by smeared out charges of some small radius $\eta$ smaller than all the inter-particle distances 
Newton's theorem says that all the particle-particle energies are preserved, but the particle-background 
interaction only increases (decreases in numerical size, but this energy is negative). 
Writing $\chi_\eta = \frac{1}{\pi \eta^2} \mathbbm{1}_{B(0, \eta)}$ we thus have
\[
\begin{aligned}
& {} \mathcal{E}_\text{Jel}(\Omega_N + \tau, x_1, \ldots, x_N)
	\\
	& \quad \leq 
		- \sum_{j<k} \iint \log |(x_j + y) - (x_k + z)| \chi_\eta(y)\chi_\eta(z) \ud y \ud z
	\\
	& \qquad 
		+ \sum_{j=1}^N \int_{\Omega_N + \tau}\int \log |(x_j + z) - y| \chi_\eta(z) \ud z \ud y
		- \frac{1}{2} \iint_{(\Omega_N + \tau) \times (\Omega_N + \tau)} \log|x-y| \ud x\ud y
	\\
	& \quad = 
		\sum_{j < k} 2D(\chi_\eta(\cdot - x_j), \chi_\eta(\cdot - x_k))
		-\sum_{j = 1}^N 2D(\chi_\eta(\cdot - x_j), \mathbbm{1}_{\Omega + \tau})
		+ D(\mathbbm{1}_{\Omega_N + \tau})
	\\
	& \quad = 
		D\left(\sum_{j=1}^N \chi_\eta(\cdot - x_j) - \mathbbm{1}_{\Omega_N + \tau}\right)
		-N\left(D(\chi_1) - \frac{1}{2} \log \eta\right).
\end{aligned}
\]
We now investigate the first term more closely. 
We may write
\[
	\sum_{j=1}^N \chi_\eta(\cdot - x_j) - \mathbbm{1}_{\Omega_N + \tau} = \sum_{\substack{k \in \Z^2 \\ |k_1|, |k_2| \leq K}} f(\cdot + \ell k),
\]
where $f = \sum_{j=1}^n \chi_\eta(\cdot - x_j) - \mathbbm{1}_{C_\ell + \tau}$.
Thus
\begin{align*}
D\left(\sum_{j=1}^N \chi_\eta(\cdot - x_j) - \mathbbm{1}_{\Omega_N + \tau}\right)
	& = 
		D\left(\sum_k f(\cdot + \ell k)\right)
	= \pi \int \frac{|\hat f(p)|^2}{p^2} \abs{\sum_k e^{ipk\ell}}^2 \ud p.
\end{align*}
Since $f$ is of compact support and satisfies $\int f \ud x = 0$ and $\int xf(x) \ud x = 0$ (this is where the zero dipole moment is used)
we have that $\hat f$ is smooth and satisfies $\hat f(p) = o(p)$ as $p \to 0$, thus $\frac{|\hat f(p)|^2}{p^2}$ vanishes at zero. 
Thus, we need to consider the behaviour of $\abs{\sum_k e^{ipk\ell}}^2$ in the limit $K \to \infty$ (i.e. $N \to \infty$). 

We have
\[
	\frac{1}{(2K+1)^2}\abs{\sum_{\substack{k \in \Z^2 \\ |k_1|, |k_2| \leq K}} e^{ipk\ell}} = \prod_{\nu = 1}^2 \frac{\sin^2(\ell p_\nu (K + 1/2))}{\sin^2( \ell p_\nu / 2)}
	\rightharpoonup \left(\frac{2\pi}{\ell}\right)^2 \sum_{p\in \frac{2\pi}{\ell} \Z^2} \delta_p
\]
weakly, and so 
\begin{eqnarray*}
\frac{1}{N} D\left(\sum_k f(\cdot + \ell k)\right)
	& \overset{N\to \infty}{\longrightarrow} &
		\frac{\pi}{\ell^2} \int \frac{|\hat f(p)|^2}{p^2} \left(\frac{2\pi}{\ell}\right)^2 \sum_{p\in \frac{2\pi}{\ell} \Z^2} \delta_p \ud p 
	\\
	& = &
		\frac{\pi}{\ell^4} \sum_{p} \int_{C_\ell + \tau} \int_{C_\ell + \tau} f(x) e^{ipx} f(y) e^{-ipy} \frac{1}{p^2} \ud x \ud y
	\\
	& = &
		\frac{1}{2n} \int_{C_\ell + \tau} \int_{C_\ell + \tau} G_\ell(x-y) f(x) f(y) \ud x \ud y.
\end{eqnarray*}
Plugging in the definition of $f$ and using that $\int_{C_\ell} G_\ell \ud x = 0$ we thus have
\begin{multline*}
\frac{1}{2n} \int_{C_\ell + \tau} \int_{C_\ell + \tau} G_\ell(x-y) f(x) f(y) \ud x \ud y
	\\ = 
		\frac{1}{2} \iint G_\ell(x-y) \chi_\eta(x)\chi_\eta(y) \ud x \ud y	
		+ \frac{1}{n}\sum_{j < k} \iint G_\ell(x-y) \chi_\eta(x - x_j) \chi_\eta(y - x_k) \ud x \ud y.
\end{multline*}
Since $\chi_\eta \rightharpoonup \delta$ as $\eta \to 0$ the second term converges to $\frac{1}{n}\sum_{j < k} G_\ell(x_j - x_k)$.
This is exactly the first term in the functional $\mathcal{E}_{\text{per}, \ell}$ as desired.
We now deal with the other term in the limit $\eta \to 0$.
\begin{align*}
& \frac{1}{2} \iint G_\ell(x-y) \chi_\eta(x)\chi_\eta(y) \ud x \ud y
	\\
	& \qquad = 
		\frac{1}{2} \iint \left(- \log \eta - \log |x-y| + \log \ell + C_\text{mad} + o_{\eta \to 0}(1)\right) \chi_1(x)\chi_1(y) \ud x \ud y,
\end{align*} 
where $o_{\eta\to 0}(1)$ vanishes as $\eta\to 0$ uniformly in $x,y$ by the compact support of $\chi_1$.
Hence
\[
	\frac{1}{2} \iint G_\ell(x-y) \chi_\eta(x)\chi_\eta(y) \ud x \ud y
	= -\frac{1}{2}\log \eta + D(\chi_1) + \frac{1}{2}\log \ell + \frac{C_\text{mad}}{2} + o_{\eta \to 0}(1).
\]
Putting everything together we now conclude
\begin{align*}
& \frac{1}{N} 
	D\left(\sum_{j=1}^N \chi_\eta(\cdot - x_j) - \mathbbm{1}_{\Omega_N + \tau}\right)
	-\left(D(\chi_1) - \frac{1}{2} \log \eta\right)
	\\
	& \qquad 
	\overset{N\to \infty}{\longrightarrow}
		\frac{1}{n}\sum_{j < k} G_\ell(x_j - x_k) 
		+ \frac{1}{2}\iint G_\ell(x-y) \chi_\eta(x)\chi_\eta(y) \ud x \ud y
		- D(\chi_1) + \frac{1}{2}\log \eta
	\\
	& \qquad
	\overset{\eta \to 0}{\longrightarrow}
		\frac{1}{n}\sum_{j < k} G_\ell(x_j - x_k) + \frac{1}{2} \log \ell + \frac{1}{2}C_\text{mad}.
\end{align*}
This proves the one inequality. To conclude the other inequality 
note that the error we made in replacing the point charges with smeared out ones 
can be bounded by 
$\int - \log |y| \chi_\eta(y) \ud y = O(-\eta^2 \log \eta)$ per particle by Newton's theorem.
This vanishes as $\eta\to 0$ and thus we conclude the equality
\[
	\lim_{N\to \infty} \frac{\mathcal{E}_\text{Jel}(\Omega_N, x_1, \ldots, x_N)}{N} = \frac{\mathcal{E}_{\text{per}, \ell}(x_1, \ldots, x_n)}{n}.
\]
This proves that $e_\text{UEG} \leq \frac{\mathcal{E}_{\text{per}, \ell}(x_1, \ldots, x_n)}{n}$.

\section{Upper Bound for the Periodic Energy}
\label{sect.upper.bound.per}
We now show that 
\[
	\limsup_{N\to \infty} \min_{x_1, \ldots, x_N \in C_L} \frac{\mathcal{E}_{\text{per}, L}(x_1, \ldots, x_N)}{N}
		\leq e_\text{Jel}.
\]
This finishes the proof of \Cref{thm.main}.
This argument is again more or less the same as in \cite{lewin.lieb.seiringer.19}.
Again, there are some slight differences in the case $d=2$,
which is why we give the argument here.

First, we show that a version of Newton's theorem hold for the periodic potential $G_\ell$. 
Namely that separated neutral radial charge densities have zero total interaction. 
More precisely, let $\rho$ be any compactly supported radial neutral charge distribution,
i.e. $\supp \rho \subset B_R$ for some $R > 0$, $\rho$ is radial and $\int \rho \ud x = 0$.
Let $L$ be large enough so that $B_R \subset C_L$.
Then we have that $V = \sum_{k \in \Z^2} (\rho * -\log)(\cdot + L k)$ satisfies
\[
	-\Delta V = 2\pi \sum_k \rho(\cdot + Lk) = \rho * \left(2\pi \left(\sum_k \delta_{Lk} - \frac{1}{L^2}\right)\right) = -\Delta (\rho * G_L).
\]
(Note the importance of $\rho$ being neutral, so $\rho * 1 = 0$.)
Since both $V$ and $\rho * G_L$ are $C_L$-periodic this shows that they differ by a 
periodic harmonic function, i.e. a constant. 
Moreover, by Newton's theorem we have that $V$ vanish on $C_L\setminus B_R$, thus we see that $\rho * G_L$ 
is constant on $C_L\setminus B_R$, and so for another neutral radial charge distribution $\rho'$ supported in this region
their interaction vanish,
$\iint G_L(x-y) \rho(x) \rho'(y) \ud x \ud y = 0$.\footnote{Note that
the analogous statement in \cite{lewin.lieb.seiringer.19} is wrong. 
There it is claimed that $\rho * G_L = 0$ on $C_L\setminus B_R$. 
This is not true in general. In general we only have $\rho * G_L$ constant on $C_L \setminus B_R$.
Since the result is only ever used for interactions between neutral charges (as we do here), their proof that 
$e_\text{UEG} = e_\text{per} = e_\text{Jel}$ in dimension $d=3$ still works.
}

We use the Swiss cheese theorem \cite[sect. 14.5]{lieb-seiringer} to fill (most of) the cube $C_L$ 
with balls of integer volume ranging in sizes from some fixed $\ell_0$ to a largest size of order $\ell$. 
The ratio of the volume not covered by the balls is small in comparison to the volume of the cube, 
in the sense that if we take $\ell \to \infty$ after taking $L\to \infty$ this ratio vanishes.

We now construct a trial state using these balls. In each ball $B_n$ we place $N_n = |B_n|$ 
particles in the optimal Jellium configuration for the ball $B_n$.
(Note that $n$ refers to the index of the ball, and not its radius.)
The remaining $M = N - \sum_n |B_n|$ particles are placed uniformly in the remainder $S = C_L \setminus \bigcup_n B_n$, 
meaning that we smear the particles out in this region. 
This yields
\[
\begin{aligned}
\min \mathcal{E}_{\text{per}, L}(x_1, \ldots, x_N)
	& \leq 
		\sum_{1\leq j < k \leq N-M}G_L(x_j - x_k)
		+\sum_{j = 1}^{N - M} \int_S G_L(x_j - y) \ud y
	\\
	& \quad		
		+\frac{1}{2}\left(1 - \frac{1}{M}\right) \iint_{S\times S} G_L(x-y)\ud x\ud y
		+ \frac{N}{2}\left(\log L + C_\text{mad}\right)
	\\
	& = 
		\sum_{1\leq j < k \leq N-M}G_L(x_j - x_k)
		- \sum_n \sum_{j = 1}^{N - M} \int_{B_n} G_L(x_j - y) \ud y
	\\
	& \quad
		+ \frac{1}{2}\left(1 - \frac{1}{M}\right) \sum_{n,m} \iint_{B_n\times B_m} G_L(x-y)\ud x\ud y
		+ \frac{N}{2}\left(\log L + C_\text{mad}\right), 
\end{aligned}
\]
where $x_1, \ldots, x_{N-M}$ denote the points in $\bigcup_n B_n$
and we used that $\int_{C_L} G_L \ud x = 0$. 
Now, by rotating the charges inside each of the balls separately and taking the average over all such rotations
we may use the modified Newton's theorem above to conclude that the balls don't interact with each other.
Writing $\tilde G_L$ for the rotational average of $G_L$ we thus get the upper bound
\begin{multline*}
	\sum_n 
		\left(
			\sum_{1\leq j < k \leq |B_n|} \tilde G_L \left(x_j^{(n)} - x_k^{(n)}\right) - \sum_{j=1}^{|B_n|} \int_{B_n} \tilde G_L \left(x_j^{(n)} - y\right) \ud y
		\right.
		\\
		\left.
			+ \frac{1}{2} \left(1 - \frac{1}{M}\right)\iint_{B_n \times B_n} \tilde G_L(x-y) \ud x \ud y
			+ \frac{|B_n|}{2}\left(\log L + C_\text{mad}\right)
		\right).
\end{multline*}
As $L\to \infty$ we have that $G_L(x) = -\log |x| + \log L + C_\text{mad} + o(1)$.
Plugging this into the bound above the $\frac{1}{M}$-term is $o(1)$ and the remaining
$\log L$ and $C_\text{mad}$-terms cancel.
What we are left with is the bound
\[
	\sum_n \left(\mathcal{E}_\text{Jel}\left(B_n, x_1^{(n)}, \ldots, x_{|B_n|}^{(n)}\right) + o_{L\to \infty}(1)\right),
\]
where the term $o_{L\to \infty}(1)$ depends on $\ell$, but not on $n$.
Dividing by $N=L^2$ and taking the consecutive limits $L\to \infty$, $\ell\to \infty$ and $\ell_0\to \infty$
this gives $e_\text{Jel}$, by the existence of the thermodynamic limit for Jellium \cite{sari.merlini.76}. 
We conclude that
\[
	\limsup_{N\to \infty} \min_{x_1, \ldots, x_N \in C_L} \frac{\mathcal{E}_{\text{per}, L}(x_1, \ldots, x_N)}{N}
		\leq e_\text{Jel}.
\]
Hence we have shown
\[
	e_{\text{UEG}} 
		\leq \liminf_{N\to \infty} \min_{x_1, \ldots, x_N \in C_L} \frac{\mathcal{E}_{\text{per}, L}(x_1, \ldots, x_N)}{N}
		\leq \limsup_{N\to \infty} \min_{x_1, \ldots, x_N \in C_L} \frac{\mathcal{E}_{\text{per}, L}(x_1, \ldots, x_N)}{N}
		\leq e_\text{Jel}.
\]
Thus $e_\text{UEG} = e_\text{per} = e_\text{Jel}$.

\section{Proof of Theorem 3.1}
\label{sect.lattice}
We now give the proof of \Cref{thm.lattice.jellium}.
\begin{proof}[Proof of \Cref{thm.lattice.jellium}]
First, we extend $V_s$ to complex-valued $s$ as follows.
For $s\in \C\setminus \R$ we define $V_s(x) = |x|^{-s}$.
Now define the functions $W_s$ and $\tilde W_s$ for any $s\in \C$ by
\[
W_s = V_s - 2 V_s * \mathbbm{1}_Q + V_s * \mathbbm{1}_Q * \mathbbm{1}_Q = V_s * \left(\delta - \mathbbm{1}_Q\right) * \left(\delta - \mathbbm{1}_Q\right),
\quad
\text{and}
\quad
\tilde W_s = |\cdot|^{-s}*\left(\delta - \mathbbm{1}_Q\right)*\left(\delta - \mathbbm{1}_Q\right).
\]
Note that $W_s = \tilde W_s$ if $s\in \C \setminus(-\infty,0]$, that $W_s = - \tilde W_s$ if $s < 0$,
and that $W_s = \left.\dd{s}\tilde W_s\right\vert_{s=0}$ if $s = 0$.

By a tedious but straightforward Taylor expansion one checks that $\tilde W_s(x), W_s(x) \sim |x|^{-\tRe(s)-4}$ for large $|x|$.
Let $x_1, \ldots, x_N$ be $N$ points on the lattice, say $\{x_1,\ldots,x_N\} = \mathcal{L} \cap B_R$
for some $R$,
 and set $\Omega_N = \bigcup_{j=1}^N \left(Q + x_j\right)$. 
Consider now $s < d$. Then
\begin{eqnarray*}
	\sum_{j < k} W_s(x_j - x_k)
		& = &
			\sum_{j <k } V_s(x_j - x_k)
			- 2 \sum_{j < k} \int_Q V_s(x_j - x_k - y) \ud y
		\\
		& & 
			+ \sum_{j < k} \iint_{Q \times Q} V_s(x_j - x_k - y - z) \ud y \ud z
		\\
		& = &
			\sum_{j <k } V_s(x_J - x_k)
			-\left(
				\sum_{j=1}^N \int_{\Omega_N} V_s(x_j - y) \ud y
				-N \int_Q V_s(y) \ud y
			\right)
		\\
		& &
				+ \frac{1}{2}\iint_{\Omega_N \times \Omega_N} V_s(y-z) \ud y \ud z
				- \frac{1}{2} \sum_{j=1}^N \iint_{Q\times Q} V_s(y-z) \ud y \ud z
		\\
		& = &
			\mathcal{E}_{\text{Jel}, d, s}(\Omega_N, x_1, \ldots, x_N) + N\int_Q V_s(y) \ud y - \frac{N}{2}\iint_{Q\times Q} V_s(y-z) \ud y \ud z.
\end{eqnarray*}
For $s > d-4$ the sum $\sum_{x\in \mathcal{L}\setminus 0} W_s(x)$ converges 
and so we may take the thermodynamic limit.
\[
	e_{\text{Jel}, s}^\mathcal{L} = \lim_N \frac{1}{N}\sum_{j < k} W_s(x_j - x_k) - \int_Q V_s(y) \ud y + D_{d,s}(\mathbbm{1}_Q)
		= \sum_{x\in \mathcal{L}\setminus 0} W_s(x) - \int_Q V_s(y) \ud y + D_{d,s}(\mathbbm{1}_Q).
\]
Let now $s\in \C, \text{Re}(s) > d$. We may then write 
\begin{align*}
\zeta_{\mathcal{L}}(s) 
	& = 
		\frac{1}{2} \sum_{x\in \mathcal{L}\setminus 0} \frac{1}{|x|^s}
	\\
	& = 
		\frac{1}{2} \sum_{x\in \mathcal{L}\setminus 0} W_s(x)
		+ \sum_{x\in \mathcal{L}\setminus 0} \int_Q V_s(x - y) \ud y
		-\frac{1}{2} \sum_{x\in \mathcal{L}\setminus 0} \iint_{Q\times Q} V_s(x-y-z) \ud y \ud z
	\\
	& = 
		\frac{1}{2} \sum_{x\in \mathcal{L}\setminus 0} W_s(x) 
		+ \int_{\R\setminus Q} V_s(y) \ud y
		- \frac{1}{2} \int_{\R\setminus Q}\int_Q V_s(y-z) \ud y \ud z.
\end{align*}
We now want to write this in a form, that makes sense for all $s\ne d$ satisfying $\tRe(s) > d-4$.

For any fixed $\eps > 0$ such that $B_\eps \subset Q$ we have
\begin{align*}
\int_{\R \setminus Q} V_s(y) \ud y 
	& = 
		\int_{|y| > \eps} V_s(y) \ud y - \int_{Q\setminus B_\eps} V_s(y) \ud y
	\\
	& = 
		\abs{S^{d-1}}\eps^{d-s} \frac{1}{s-d} - \int_{Q\setminus B_\eps} \frac{1}{|y|^s} \ud y.
\end{align*}
Both of these terms are holomorphic for $s \ne d$ for any fixed $\eps > 0$.
For $\tRe(s) < d$ we may take $\eps \to 0$. Hence the analytic continuation of this term is
$-\int_Q |y|^{-s} \ud y$ when $\tRe(s) < d$. 

For the second term we write
\[
	\int_{\R \setminus Q} \int_Q \frac{1}{|y-z|^s} \ud y \ud z = \int_Q \int_{\R \setminus (Q + z)} \frac{1}{|w|^s} \ud w \ud z.
\]
We now split this integral according to
\begin{align*}
	\int_Q \int_{\R \setminus (Q + z)} 
		& =
			\int_Q \int_{|w| > \eps |z|} - \int_Q \int_{(Q+z)\setminus B_{\eps |z|}} 
		\\
		& = 
			\int_{Q\setminus B_\delta}\int_{|w| > \eps |z|} 
			+ \int_{B_\delta}\int_{|w| > \eps |z|}
			- \int_{Q\setminus B_\delta} \int_{(Q+z)\setminus B_{\eps|z|}}
			- \int_{B_\delta} \int_{(Q+z)\setminus B_\rho}
			- \int_{B_\delta} \int_{B_\rho \setminus B_{\eps|z|}}
\end{align*}
where $\eps, \delta, \rho > 0$ are all sufficiently small.
The terms 
\[
	\int_{Q\setminus B_\delta} \int_{(Q+z)\setminus B_{\eps|z|}} |w|^{-s} \ud w \ud z, 
	\qquad 
	\int_{B_\delta} \int_{(Q+z)\setminus B_\rho} |w|^{-s} \ud w \ud z
\]
are analytic in $s$. For the remaining terms we calculate.
The term
\[
	\int_{Q\setminus B_\delta}\int_{|w| > \eps |z|} |w|^{-s} \ud w \ud z
		= \abs{S^{d-1}}\frac{1}{s-d}\eps^{d-s}\int_{Q \setminus B_\delta} |z|^{d-s} \ud z
\]
makes sense for $\tRe(s) > d$ and extends analytically to $s \ne d$.
The term
\[
	\int_{B_\delta}\int_{|w| > \eps |z|} |w|^{-s} \ud w\ud z
	= \abs{S^{d-1}}^2 \frac{1}{s-d}\frac{1}{2d-s} \eps^{d-s} \delta^{2d -s}
\]
makes sense for $d < \tRe(s) < 2d$ and extends analytically to $s\ne d, 2d$.
The term
\[
	 \int_{B_\delta} \int_{B_\rho \setminus B_{\eps|z|}} |w|^{-s} \ud w\ud z
	 = \abs{S^{d-1}} \abs{B_1} \frac{1}{d-s}\delta^d \rho^{d-s} 
	 	- \abs{S^{d-1}}^2 \frac{1}{d-s}\eps^{d-s} \frac{1}{2d-s}\delta^{2d-s}
\]
makes sense for $d < \tRe(s) < 2d$ and extends analytically to $s\ne d, 2d$.
The ``poles'' at $s = 2d$ in fact cancel out, so $2d$ is not a pole of $\zeta_\mathcal{L}(s)$.

For $\tRe(s) < d$ we may take $\eps, \delta, \rho \to 0$ in a suitable order. 
All these terms combined then give the limit 
\[
	-\int_Q \int_{Q+z} |w|^{-s} \ud w \ud z = - \iint_{Q\times Q} \frac{1}{|w-z|^s} \ud w \ud z.
\]
Thus, for $d-4 < \tRe(s) < d$ we have that (for the analytic continuation)
\[
\zeta_\mathcal{L}(s) = \frac{1}{2}\sum_{x\in \mathcal{L}\setminus0} \tilde W_s(x) - \int_Q \frac{1}{|y|^s} + \frac{1}{2}\iint_{Q\times Q} \frac{1}{|y-z|^s} \ud y \ud z.
\]
Thus for real $d-4 < s < d$ with $s \ne 0$ we have
\[
\zeta_\mathcal{L}(s)
	= \begin{cases}
	e_{\text{Jel}, s}^\mathcal{L} & \text{if } s > 0, \\
	- e_{\text{Jel}, s}^\mathcal{L} & \text{if } s < 0.
	\end{cases}
\]
For $s=0$ we have $W_{s=0} = \left.\dd{s} \tilde W_s \right\vert_{s=0}$ and similarly 
$V_{s=0} = \left.\dd{s} |\cdot|^{-s} \right\vert_{s=0}$. Thus 
\begin{multline*}
	e_{\text{Jel}, s=0}^\mathcal{L}
		= \sum_{x\in \mathcal{L}\setminus 0} W_{s=0}(x) - \int_Q V_{s=0}(y) \ud y + D_{d,s=0}(\mathbbm{1}_Q)
	\\
		= \dd{s} 
				\left[
			\sum_{x\in \mathcal{L}\setminus 0} \tilde W_{s}(x) 
			- \int_Q \frac{1}{|y|^s} \ud y 
			+ \frac{1}{2}\iint_{Q\times Q} \frac{1}{|y-z|^s} \ud y \ud z
				\right]_{s=0}
		= \zeta_\mathcal{L}'(0).
\end{multline*}
This finishes the proof of \Cref{thm.lattice.jellium}.
\end{proof}

\section*{Acknowledgments}
The author would like to thank Robert Seiringer 
for guidance and many helpful comments on this project.
Additionally the author would like to thank 
Mathieu Lewin for his comments on the manuscript
and 
Lorenzo Portinale 
for providing his lecture notes for the course ``Mathematics of quantum many-body systems''
in spring 2020,
taught by Robert Seiringer. The proof of \Cref{thm.lattice.jellium}
is inspired by these lecture notes.

\appendix
\section{Proof of Proposition 5.5}
\begin{proof}[Proof of \Cref{prop.d.convolve.dipole.=0}]
Let $L = |A|^{1/2}$ be the side length of $A$, find an $\ell' \geq 4\ell$ of order 1 such that $L/\ell'$ is an integer.
Let $Q$ denote the square of side length $\ell'$ centered at zero. Tile the plane with translates of $Q$ 
such that, for the relevant translates, the centers $y_j$ lie on the boundary $\partial A$. 
That is,
$\R^2 = \bigcup_j (y_j + Q)$ and if $y_j + Q$ intersect both $A$ and $A^c$, then $y_j \in \partial A$.
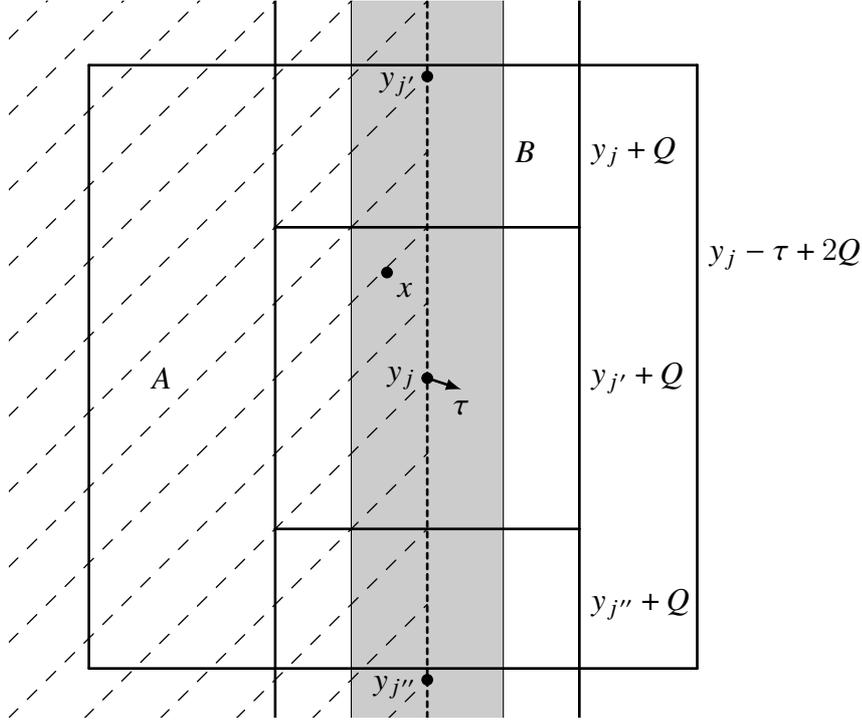
\begin{figure}[htb]
\center
\definecolor{light grey}{rgb}{0.8,0.8, 0.8}
\begin{tikzpicture}[line cap=round,line join=round,>=triangle 45,x=1 cm,y=1 cm]
\clip(-5.5,-4.5) rectangle (7.,5.);
\fill[line width=0.1 pt,fill=light grey] (-1.,7.) -- (1.,7.) -- (1.,-7.) -- (-1.,-7.) -- cycle;
\foreach \x in {-14, ..., 14}{
	\draw[line width = 0.5 pt, dash pattern = on 5 pt off 5 pt] (-7, \x) -- (0, \x + 7);
}
\draw [line width = 1 pt,dash pattern=on 2 pt off 2 pt] (0.,-4.5) -- (0.,5.);
\draw [line width = 1 pt] (2.,-2.)-- (2.,2.);
\draw [line width = 1 pt] (2.,2.)-- (-2.,2.);
\draw [line width = 1 pt] (-2.,2.)-- (-2.,-2.);
\draw [line width = 1 pt] (-2,-2) -- (2, -2);
\draw [line width = 1 pt] 
	(-2,-2) -- (-2, -6)
	(-2, -6) -- (2, -6)
	(2, -6) -- (2, -2)
	(2,2) -- (2,6)
	(2,6) -- (-2, 6)
	(-2,6) -- (-2,2);
\draw [line width = 1 pt] (3.55,-3.85)-- (3.55,4.15);
\draw [line width = 1 pt] (3.55,4.15)-- (-4.45,4.15);
\draw [line width = 1 pt] (-4.45,4.15)-- (-4.45,-3.85);
\draw [line width = 1 pt] (-4.45, -3.85) -- (3.55, -3.85);
\draw [line width=0.4 pt] (-1.,7.)-- (1.,7.);
\draw [line width=0.4 pt] (1.,7.)-- (1.,-7.);
\draw [line width=0.4 pt] (1.,-7.)-- (-1.,-7.);
\draw [line width=0.4 pt] (-1.,-7.)-- (-1.,7.);
\draw [-{latex[length = 1mm]}, line width = 1 pt] (0,0) -- (0.45, -0.15);
\draw (2,0) node[anchor= west] {$y_{j'} + Q$};
\draw (2,3) node[anchor= west] {$y_j + Q$};
\draw (2,-3) node[anchor= west] {$y_{j''} + Q$};
\draw (3.55, 2) node[anchor=north west] {$y_j - \tau + 2Q$};
\draw (1,3) node[anchor= west] {$B$};
\draw (0.45,-0.15) node[anchor=north] {$\tau$};
\draw (-0.53,1.4) node[anchor=north west] {$x$};
\draw (0,0) node[anchor=east] {$y_j$};
\draw (0,3.9) node[anchor=east] {$y_{j'}$};
\draw (0,-4.1) node[anchor=east] {$y_{j''}$};
\draw (-3.5,0) node {$A$};
\begin{scriptsize}
\draw [fill=black] (0.,0.) circle (2 pt);
\draw [fill = black] (0, -4) circle (2 pt);
\draw [fill = black] (0, 4) circle (2 pt);
\draw [fill=black] (-0.53,1.4) circle (2 pt);
\end{scriptsize}
\end{tikzpicture}
\caption{Picture of the boundary region $B$ (in grey) with the relevant translates of $Q$.}
\label{figure.boundary.A.construction}
\end{figure}
Now, for any $x\in B$ we have that $x \in y_j + Q$ for some (unique) $y_j \in \partial A$, see \Cref{figure.boundary.A.construction}. Thus
\[
	D\left(\mathbbm{1}_{A + \tau} - \mathbbm{1}_A * \mu, b\right)
		= - \frac{1}{2}\sum_{j: y_j \in \partial A} 
			\int_{y_j + Q} b(x) \int \log|x-y| \left(\mathbbm{1}_{A+\tau} - \mathbbm{1}_A * \mu\right)(y) \ud y \ud x
\]
We now split the $y$-integral in two according to whether $y$ is ``close'' to $x$, namely if $y \in y_j + 2Q$ or if $y$ is  ``far'' from $x$, namely if $y\notin y_j + 2Q$. 
For the close $y$'s we get the contribution
\[
	- \frac{1}{2}\sum_{j: y_j \in \partial A} 
		 \int_{y_j + Q} b(x) \int_{y_j + 2Q} \log|x-y| \left(\mathbbm{1}_{A+\tau} - \mathbbm{1}_A * \mu\right)(y) \ud y \ud x
		= O\left(N^{1/2}\right)
\]
since there are $O(N^{1/2})$ many such $j$'s, with each some order 1 contribution. For the $y$'s far away we get the contribution 
\[
	- \frac{1}{2}\sum_{j: y_j \in \partial A} \int_{y_j + Q} b(x) \int_{(y_j + 2Q)^c} \log|x-y|
		\left[\mathbbm{1}_{A+\tau}(y) - \int \mathbbm{1}_A(y-a) \ud \mu(a)\right] \ud y \ud x
\]
we compute
\begin{eqnarray*}
	& = & 
		- \frac{1}{2} \sum_j \int_{y_j + Q} b(x) \int 
			\left[
				\int_{(y_j - \tau + 2Q)^c \cap A} \log|x-y - \tau| \ud y 
			\right.
	\\
	&&
			\left.
			\phantom{\frac{1}{2} \sum_j \int_{y_j + Q} b(x) \int}
				- \int_{(y_j - a + 2Q)^c\cap A}  \log|x-z-a| \ud z
			\right] \ud \mu(a) \ud x
	\\
	& = &
		- \frac{1}{2} \sum_j \int_{y_j + Q} b(x) \int 
			\int_{(y_j - \tau + 2Q)^c \cap A} \log|x-y - \tau| - \log|x-y-a| \ud y \ud \mu(a) \ud x
	\\
	& & 
		- \frac{1}{2} \sum_j \int_{y_j + Q} b(x) \int
			\left[\int_{(y_j - \tau +2Q)^c\cap A} - \int_{(y_j - a + 2Q)^c\cap A}\right]  \log|x-y-a| \ud y \ud \mu(a) \ud x.
\end{eqnarray*}
The first term here, we again use the Taylor expansion of $\log$. we thus get for the integrand in the $x$-integral
\begin{multline*}
	b(x) \int \int_{(y_j - \tau + 2Q)^c \cap A}
		\frac{(x-y-\tau) \cdot (\tau - a)}{|x-y-\tau|^2} + O\left(\frac{1}{|x-y|^2}\right) \ud y \ud \mu(a)
	\\
		= O\left(\int_{(y_j - \tau + 2Q)^c\cap A} \frac{1}{|x-y|^2} \ud y\right).
\end{multline*}
Thus, computing the $x$-integral of this we get a term which is $O(N^{1/2}\log N)$ by \Cref{prop.n.1/2.logn}
(note that for $y\in (y_j - \tau  + 2Q)^c \cap A$ and $x\in y_j + Q$ we have that $|x-y| \geq \ell$).
For the second term, the $x$-integrand is
\begin{multline*}
	- b(x) \int \left[\int_{(y_j - \tau + 2Q)^c\cap A} - \int_{(y_j - a + 2Q)^c \cap A}\right] \log|x-y-a| \ud y \ud \mu(a)
	\\ 
	= - b(x) \int \left[\int_{(y_j - a + 2Q) \cap A} - \int_{(y_j - \tau + 2Q) \cap A} \right] \log|x-y-a| \ud y \ud \mu(a)
\end{multline*}
This is only an integral of $y$'s ``close'' to $x$, and so a similar argument as above gives that 
when we integrate this over all $x$ we get a term which is $O(N^{1/2})$. 
We conclude the desired
\[
	D\left(\mathbbm{1}_{A + \tau} - \mathbbm{1}_A * \mu, b\right)
	= O\left(N^{1/2}\log N\right).
	\qedhere
\]
\end{proof}

\section{Lower Bound of the Jellium Energy}
We here present the proof of the lower bound of the Jellium energy from \cite{lieb.narnhofer.75,sari.merlini.76}.
\begin{prop}[{\cite{lieb.narnhofer.75,sari.merlini.76}}]
Let $\Omega_N$ be any domain with $|\Omega_N| = N$ and let $x_1,\ldots,x_N \in \Omega_N$ 
be any configuration of points. Then
\[
	\mathcal{E}_\textnormal{Jel}(\Omega_N, x_1,\ldots,x_N) \geq - \left(\frac{3}{8} + \frac{1}{4} \log \pi\right) N \simeq -0.66118 \, N.
\]
In particular $e_\textnormal{Jel} \geq - \left(\frac{3}{8} + \frac{1}{4} \log \pi\right) \simeq -0.66118.$
\end{prop}
\begin{proof}
The idea is to smear out the electrons to a ball of radius $a$ of uniform charge.
Then optimise the result over the radius $a$. 
Define
\begin{align*}
U_{BB} 
	& := - \frac{1}{2} \iint_{\Omega_N \times \Omega_N} \log|x-y| \ud x \ud y, 	
		&& \text{the background self-energy},
	\\
U_j  
	& := \int_{\Omega_N} \log|x_j - y| \ud y,
		&& \text{particle $j$ - background interaction},	
	\\
U_{jk}
	& := - \log|x_j - x_k|,
		&& \text{particle $j$ - particle $k$ interaction},
	\\
\hat U_j
	& := \frac{1}{|B_a|} \int_{B(x_j, a)} \int_{\Omega_N} \log|x-y| \ud x \ud y,
		&& \text{ball $j$ - background interaction},
	\\
\hat U_{jk}
	& := - \frac{1}{|B_a|^2} \int_{B(x_j, a)} \int_{B(x_k, a)} \log|x-y| \ud x\ud y,
		&& \text{ball $j$ - ball $k$ interaction}.
\end{align*}
Then $\hat U_{jj}$ is twice the self-energy of ball $j$.
We then write
\[
	\mathcal{E}_\text{Jel} (\Omega_N, x_1, \ldots, x_N)
		= 	
			\underbrace{U_{BB} + \sum_{j=1}^N \hat U_j + \frac{1}{2}\sum_{j,k} \hat U_{jk}}_{(\alpha)}
			+ \underbrace{\sum_{j=1}^N U_j - \hat U_j}_{(\beta)}
			+ \underbrace{\frac{-1}{2}\sum_{j=1}^N \hat U_{jj}}_{(\gamma)}
			+ \underbrace{\sum_{j < k} U_{jk} - \hat U_{jk}}_{(\delta)}.
\]
Now, $(\alpha)$ is the total electrostatic energy of the combined charge distribution of the smeared out electrons and the background,
i.e. $(\alpha) = D\left(\sum_{j=1}^N \frac{1}{|B_a|} \mathbbm{1}_{B(x_j, a)} - \mathbbm{1}_\Omega\right)$.
Since the entire configuration is neutral, we have $(\alpha) \geq 0$ by \Cref{prop.d(f).geq.0}.
Also $(\delta)\geq 0$, since if the balls are not overlapping, then this term is 0, but if they are overlapping, 
then by Newton's theorem this term is positive.
Now, $(\beta)$ can be bounded by Newton's theorem
\begin{align*}
U_j - \hat U_j
	& = \int_{\Omega_N} \log|x_j - z| - \frac{1}{|B_a|}\int_{B(x_j, a)} \log|x-z| \ud x \ud z
	\\
	& \geq \frac{1}{|B_a|} \int_{B(x_j, a)}\int_{B(x_j, a)} \log|x_j - z| - \log |x-z| \ud x \ud z
	\\
	& = \frac{1}{|B_a|} \iint_{B_a \times B_a} \log|z| - \log|x-z| \ud x \ud z.
\end{align*} 
We have equality if $B(x_j, a)\subset \Omega_N$, but in general always the stated inequality.
Lastly, $(\gamma)$ is given by $(\gamma) = \frac{N}{2|B_a|^2}\iint_{B_a \times B_a} \log|x-y| \ud x \ud y$.

Computing $(\gamma)$ and the bound for $(\beta)$ we arrive at
$\mathcal{E}_{\text{Jel}}	\geq \left(\frac{1}{2}\log a - \frac{1}{8} - \frac{\pi}{4}a^2\right)N$.
By optimising over $a$ we thus get
$\mathcal{E}_{\text{Jel}}
	\geq -\left(\frac{3}{8} + \frac{1}{4}\log \pi\right) N 
$ as desired.
\end{proof}

\sloppy
\printbibliography
\end{document}